\pdfoutput=1
\documentclass{article}

\usepackage{ifthen}
\newboolean{arxiv}
\setboolean{arxiv}{true}

\usepackage{amsmath,amssymb}
\usepackage{algorithm}
\usepackage{algpseudocode}
\usepackage{enumitem}
\usepackage{booktabs}

\algrenewcommand\algorithmicindent{0.6em}

\ifthenelse{\boolean{arxiv}}{
  \newtheorem{theorem}{Theorem}
  \newtheorem{corollary}{Corollary}
  \newtheorem{constraint}{Constraint}
  \newtheorem{technique}{Technique}
  \newtheorem{implication}{Implication}
  
  \newtheorem{proof}{Proof}
}{
  \newtheorem{constraint}{Constraint}
  \newtheorem{technique}{Technique}
  \newtheorem{implication}{Implication}
}

\newcommand{\before}{\sqsubseteq_\textnormal{before}}

\title{Stop talking to me---a communication-avoiding ADER-DG
realisation\thanks{ The underlying project has received funding from the European Union's Horizon
 2020 research and innovation programme under grant agreement No 671698
 (ExaHyPE). All software is freely available from \url{www.exahype.eu}.
}} 

\author{D.E.~Charrier\footnotemark[2]
\and T.~Weinzierl\footnotemark[2]}

\ifthenelse{\boolean{arxiv}}{
  \usepackage{hyperref}
  \usepackage{graphicx}
  \usepackage{geometry}
  \geometry{left=4.1cm,textwidth=14.2cm,top=2cm,textheight=24cm}
}{

\slugger{sisc}{201x}{xx}{x}{x--x}
}

\begin{document}

\maketitle

\begin{abstract}
We present a communication- and data-sensitive formulation of
ADER-DG for hyperbolic differential equation systems.
Sensitive here has multiple flavours:
First, the formulation reduces the persistent memory footprint. 
This reduces pressure on the memory subsystem.
Second, the formulation realises the underlying predictor-corrector scheme with
single-touch semantics, i.e.~each degree of freedom is read on average only once
per time step from the main memory.
This reduces communication through the memory controllers.
Third, the formulation breaks up the tight coupling of the explicit
time stepping's algorithmic steps to mesh traversals. 
This averages out data access peaks.
Different operations and algorithmic steps are ran on different grid entities.
Finally, the formulation hides distributed memory data transfer behind the
computation aligned with the mesh traversal.
This reduces pressure on the machine interconnects.
All techniques applied by our formulation are elaborated by means of a rigorous
task formalism.
They break up ADER-DG's tight causal coupling of compute steps
and can be generalised to other predictor-corrector schemes.

\end{abstract} 

\ifthenelse{\boolean{arxiv}}{

}
{



 \footnotetext[2]{
  Department of Computer Science, Durham University, 
  Lower Mountjoy South Road, Durham DH1 3LE, United Kingdom 
  (\email{tobias.weinzierl@durham.ac.uk})
 }
}

\section{Introduction}

%
%
Hyperbolic equation systems in their first order formulation
\begin{equation}
  \frac{\partial Q}{\partial t}
 +
 \nabla\cdot F(Q)
 +
   \sum_{i=1}^d B_i(Q) \frac{\partial Q}{\partial_{x_i}}
 = 
   S(Q)
 +
    \sum \delta
    \qquad
    \mbox{with}\ 
    Q: \mathbb{R}^{d+1} \mapsto  \mathbb{R}^m
 \label{equation:introduction:PDE}
\end{equation}

\noindent
describe important wave phenomena from science and engineering. 
$d \in \{2,3\}$ is the spatial dimension supplemented by time.
For $m=1$ equation (\ref{equation:introduction:PDE}) is scalar.
In the ExaHyPE project \cite{Software:ExaHyPE} underlying this work, we focus on
astrophysical and seismic phenomena:
We want to obtain a better understanding of the long-range behaviour of
earthquakes impacting critical infrastructure---this underlies seismic risk
assessment---and we want to virtually search for gravitational waves emitted
by rotating binary neutron stars or black holes.
Advance in both fields hinges on better and faster numerical tools, such that we
can increase the level of detail of the simulation, the time spans we study and
the overall domain under examination.
Solutions to (\ref{equation:introduction:PDE}) are often characterised by features spanning
multitudes of scales in space that appear, move and disappear.
Numerical observations hence have to cover large scales both in space and
time, while an efficient numerical method to solve
(\ref{equation:introduction:PDE}) has to be accurate, it requires a mesh that is dynamically adapted
to localised features, and it requires an implementation delivering MFLOPS/s.

%
%
Explicit ADER-DG \cite{Dumbser:14:Posteriori,Titarev:02:ADER} on adaptive Cartesian meshes
derived from spacetrees \cite{Weinzierl:11:Peano,Weinzierl:17:Peano}
promisingly candidates as tool as it yields very high spatial and temporal
accuracy.
It is a discontinuous Galerkin scheme and thus fits straightforwardly to
nonconformal adaptive mesh refinement (AMR).
In ADER-DG's mindset, a weak formulation in space and time yields a
prediction of how the solution would evolve within each cell if we were
allowed to neglect the impact of neighbouring cells.
This is called predictor step. 
The arising jumps along cell interfaces then are subject to a 
Riemann solver.
In a final wrap-up, ADER-DG combines the predictor's solution with the Riemann
solve to obtain the solution at the subsequent time step.
It is an explicit predictor-corrector scheme. 
The present paper studies for the first time ADER-DG's data storage, data
movement and data exchange characteristics in detail.
We expect future machines to be equipped with enormous compute power.
Transferring data from the memory into the chip and exchanging data between
cores and nodes thus will become a severely limiting factor as these data
movements require energy \cite{Dongarra:14:ApplMathExascaleComputing}.
Our paper therefore uncovers both ADER-DG's potential and fundamental challenges
w.r.t.~supercomputing.
From hereon, it introduces
techniques that make ADER-DG implementations single-touch, i.e.~read
each unknown only once per time step.
It analyses whether and how the scheme becomes fit for exascale computing where
it is key to reduce data movements.

For our applications, $\delta: \mathbb{R}^{d+1} \to
\mathbb{R}^{m}$ in (\ref{equation:introduction:PDE}) models
earthquake point sources while non-conservative operators $B_i(Q) :
\mathbb{R}^{d+1} \to \mathbb{R}^{m\times m}$ anticipate material parameter
changes.
In astrophysics, we find
the $B_i$ modelling space-time curvature
and the conservative flux $F(Q) : \mathbb{R}^{d+1} \to \mathbb{R}^{m\times
d}$ modelling the relativistic evolution of magnetic material.
Yet, the generality of the present paper allows us to transfer all insight
to any other wave equation that can be phrased in the form of
(\ref{equation:introduction:PDE}).
The present discussion is not restricted to the two ExaHyPE application areas
only.
The rigour and generality of
the proposed analysis and techniques furthermore imply that our insights
can be transferred to many other predictor-corrector schemes;
notably Finite Volumes which is a special case of our scheme.
Finally, we detail and generalise the concept of communication-avoiding
algorithms in the present manuscript and classify aspects of ADER-DG within
this terminology.

%
%
The studied ADER-DG schemes rely on high order polynomials $p\in
\{2,\ldots,9\}$ to span the solution in space and time.
Despite the high-order, ADER-DG remains a one step
scheme---one triad of prediction, Riemann solve and correction yields the
next time step's solution---where overall high order convergence is
experimentally demonstrated \cite{Dumbser:06:ADERDG,Gassner:11:ExplicitDG}.
Once ADER-DG is combined with a Finite Volume (FV) limiter
\cite{Dumbser:14:Posteriori} it becomes robust despite the presence of shocks:
the high order representation is temporarily and locally replaced with FV on a regular Cartesian patch.
Particular characteristics of ADER-DG render it a particular promising candidate
for high-performance computing.
Yet, the very same characteristics pose realisation challenges
w.r.t.~current and future supercomputers
\cite{Dongarra:14:ApplMathExascaleComputing}.
First, ADER-DG's high order polynomials spanning both space and time imply that
each grid cell carries a significant number of degrees of freedom.
Taking a cell from the main memory and writing its updates back to memory thus
is bandwidth-demanding.
Parts of these data are even required three times per time
step though the Riemann solve and the correction typically have low arithmetic
intensity.
Second, ADER-DG's single-step character involving only one Riemann solve per
time step implies that solution and (normal) flux jump between the predicted
cell solutions have to be exchanged only once per time step.
Two strong synchronisation points per time step (the exchange of the
Riemann input data plus the eigenvalues determining the CFL
condition) plus the fact that each Riemann solve per se is computationally
cheap however render the algorithmic blueprint latency- and load
balancing-sensitive.
Finally, modifications of the grid, the spatial discretisation paradigm and the
possibility to choose a unique time step size per cell after each and every time
step fit seamlessly into the discontinuous Galerkin paradigm. 
Conforming or balanced grids \cite{Sundar:08:BalancedOctrees} and coordinated
time step sizes on cells are not inherently required, while the fallback to FV injects
robustness.
Yet, this flexibility modifies the data flow pattern all the time, and it
renders the cost per cell difficult to predict.
The three properties face a machine generation where bandwidth is a scarce
resource, rigorous synchronisation struggles to scale and runtime and latency
are sensitive to communication pathways
\cite{Dongarra:14:ApplMathExascaleComputing}.

%
%
We propose to cast ADER-DG into a task language where each
algorithmic step (predictor, Riemann, corrector) defines a set of tasks of one
type.
Temporal dependencies between the steps respective task types then are
given by ADER-DG.
Our grid instantiates the tasks plus their dependencies.
It yields a partial order.
While task-based formalisms are well-established, we are not aware of any
similar formalism of our predictor-corrector scheme that clearly distinguishes task
types, task dependencies, and task graph instantiation.
A distinction however is important:
Our grid and, hence, our task 
graph may change in each time step.
Since the partial order on the tasks does not impose a spatial order on the
cells, since an assembly of any task graph is unnecessary if a
mesh already encodes it, and since this task graph is subject to frequent changes
and thus would be expensive to assemble and maintain,
we make the grid traversal
itself issue the tasks. 
We work task assembly-free.
Our work shows that it is reasonable to shift all tasks along the simulation
time axis;
some even by half a grid sweep which makes the very first traversal trigger
only half of the actions required to complete one time step.
As long as the CFL constraints evolve smoothly, this allows us to rewrite
ADER-DG optimistically with single-touch semantics \cite{Weinzierl:17:BoxMG}.
Each cell's data are only read/written once per time step.
The approach furthermore allows us to hide all distributed memory data exchange
behind the actual grid traversal and the
other, expensive tasks.
If time steps evolve non-smoothly or the limiter kicks in,
the data access cost of our scheme double as we have to roll the solution
partially back.
Yet, they still remain better than for a
straightforward predictor-corrector implementation which maps each algorithmic step onto one
grid sweep.


The remainder is organised as follows: 
We first contextualise our work w.r.t.\, communication-avoiding algorithms 
(Sect.~\ref{section:communication-avoiding}).
Next, we revise ADER-DG briefly and clarify which mesh data
structures we rely on (Sect.~\ref{section:aderdg}), before we state our main
achievements (Sect.~\ref{section:theorems}).
In Sect.~\ref{section:tasks}, we break down the algorithm's steps into tasks
and discuss techniques how to rearrange those tasks to make them
communication-avoiding.
We continue with an analysis of the arising realisation (Sect.~\ref{section:properties}) which predicts to which degree our techniques
make the algorithm's implementation bandwidth- and memory bus-modest and thus 
fit to our notion of communication-avoiding.
We continue with experimental results before we close the
discussion in Sect.~\ref{section:conclusion}.

\section{Communication-avoiding algorithms}
\label{section:communication-avoiding}

%
%
Our contribution is the rewrite of ADER-DG in terms of a communication-avoiding
algorithm expressed in a task language. Our notion 
of com\-muni\-cation\--avoiding generalises the classic term 
(\cite{Demmel:12:CommunicationAvoiding} and references therein) as it comprises
\begin{enumerate}[leftmargin=*]
  \item the elimination of data transfer volume
  \cite{Ballard:14:CommunicationAvoidingFac,Demmel:12:CommunicationAvoiding}.
  While
  we elaborate appropriate techniques in
  \cite{Eckhardt:16:SPH,Weinzierl:17:BoxMG} and clarify in the present
  manuscript how they integrate seamlessly into our realisation, the major
  contribution here is to reduce ADER-DG's memory footprint from a space-time
  footprint into a $d$-dimensional footprint.
  \item the elimination of the data transfer frequency. We fuse two
  data exchange steps into one data exchange and restrict ourselves to
  single-touch algorithms where each data item is read and written only once \cite{Reps:15:Helmholtz,Weinzierl:17:BoxMG}. 
  Only in few cases, our approach requires redundant computations.
  \item the homogenisation of data transfer. We ensure that the demand for data
  does not fluctuate significantly over the compute time. Notably, communication
  bursts are avoided. This ensures that data transfer facilities are not idle
  over long time spans and reduces time intervals when they are
  oversubscribed.
  \item the overlapping of data transfer with communication
  \cite{Demmel:12:CommunicationAvoiding,Ghysels:13:HideLatency,Jabbar:16:CommunicationAvoiding}.
  This avoids that the code bumps into waits as incoming data
  are, in the best case, already available when they are needed.
  This property results from a reordering of algorithmic steps such that
  different arithmetic operations are localised and thus, once completed,
  all affected data can travel through interconnects until the next grid
  sweep hits a mesh cell again.
  \item the localisation of data transfer
  \cite{Jabbar:16:CommunicationAvoiding}. Explicit time stepping natively
  exchanges data only between neighbouring cells and thus tickboxes this rubric.
  Yet, our contribution also
  covers temporal data access proximity, i.e.~shared memory data transfer and
  cache effects \cite{Kowarschik:03:CacheOverview}. Activities following each
  other work on spatially near data and data is not temporarily filed in main memory.
\end{enumerate}

\noindent
We consider all communication-avoiding
techniques not only to affect distributed memory but
also to apply to memory access and multicore communication.
While we derive our communication-avoiding techniques for ADER-DG and FV,
the ideas should apply to many predictor-corrector schemes.
The above
classification is generic.

\section{ADER-DG on spacetree meshes}
\label{section:aderdg}

\begin{figure}[htb]
  \begin{center}
    \includegraphics[width=0.46\textwidth]{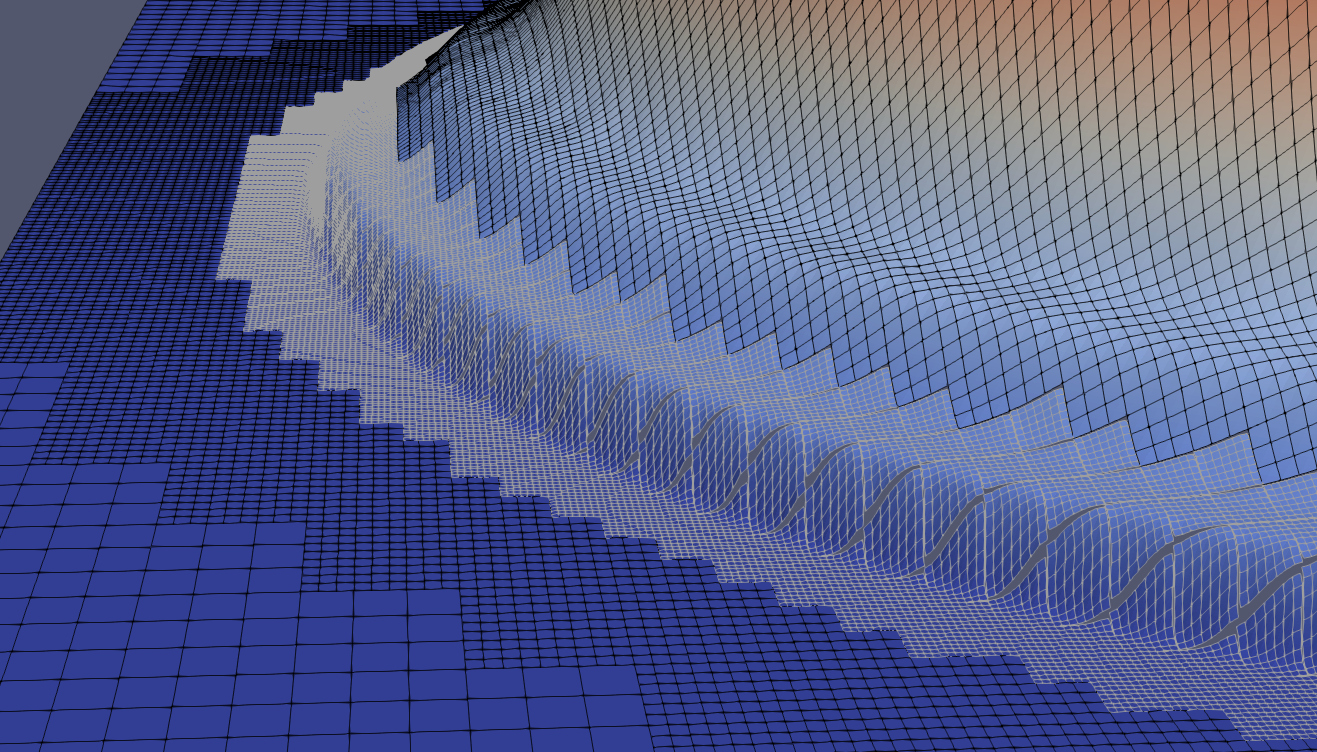}
    \includegraphics[width=0.33\textwidth]{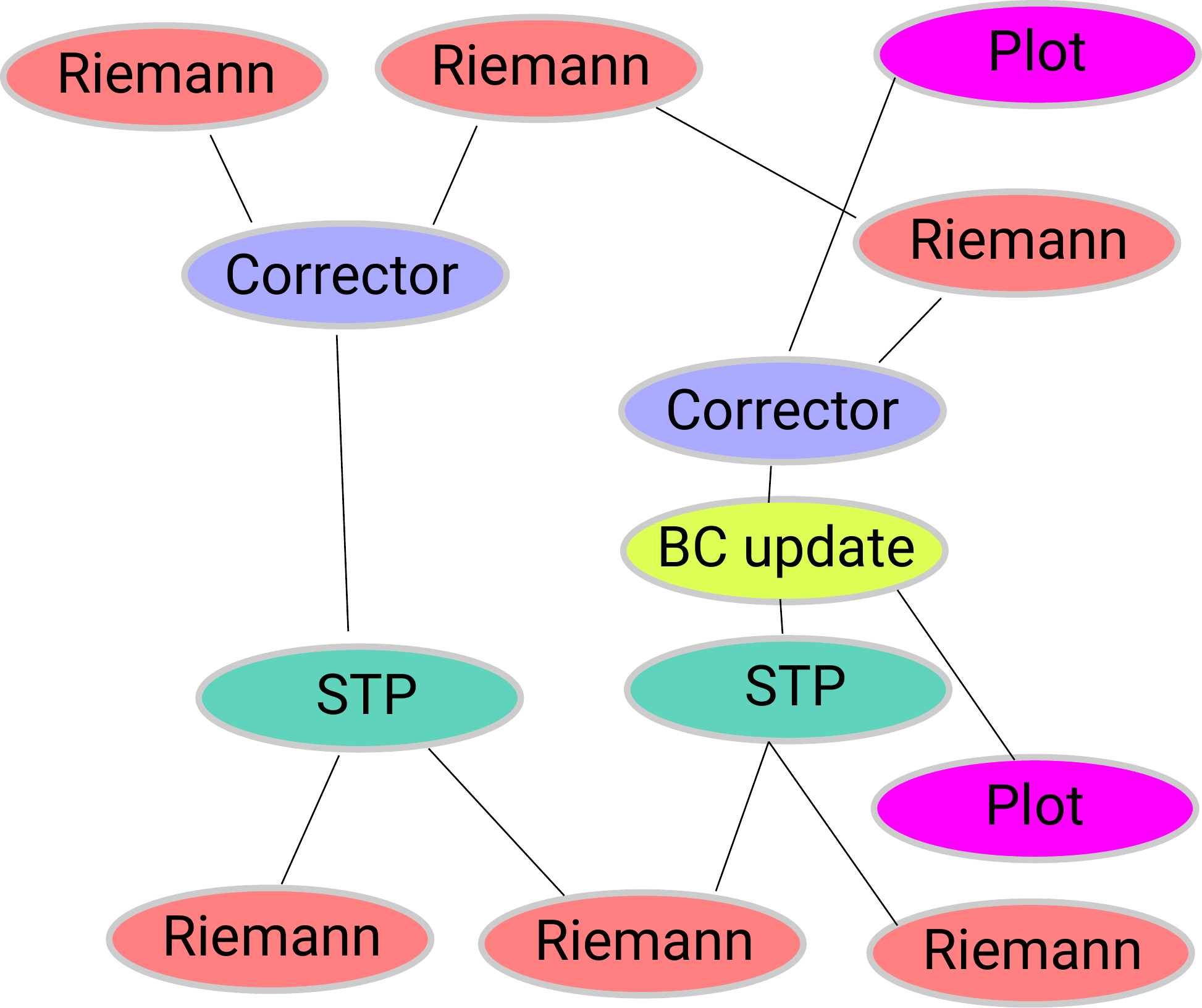}
  \end{center}
  \caption{
    Left: Adaptive Cartesian grid of a $2d$ simulation where the
    solution is extruded in the third dimension.
    Right: Characteristic ADER-DG task graph.
  }
  \label{figure:aderdg:eyecatcher}
\end{figure}

%
%
Our grids arise from a standard spacetree construction scheme based upon
tripartitioning \cite{Weinzierl:11:Peano}.
We embed the computational domain into a square or cube, respectively, and cut
this cube equidistantly into three parts.
This yields nine or 27, i.e.~$3^d$, new hypercubes with $d$ being the spatial
dimension.
We decide per hypercube whether to continue recursively, and thus end up with an
adaptive Cartesian grid (Fig.~\ref{figure:aderdg:eyecatcher}) consisting of a set $\mathbb{T}$ of squares or cubes, respectively.
We call them \textit{cells}.

%
%
Each cell $c\in\mathbb{T}$ within the adaptive Cartesian grid holds a
part of the ADER-Dis\-continuous Galerkin ($DG$) solution $Q_h(T)$.
We store $m \cdot (p+1)^{d}$ coefficients per DG cell.
They are used to cellwisely span a $p$th order polynomial modelling
all $m$ solution components from (\ref{equation:introduction:PDE}).
The polynomials are expressed via a tensor-product Lagrange basis with
Gauss-Legendre collocation points. The support of the basis
functions is limited to the interior of $c$.
Values of polynomials along cell faces, i.e.~extrapolated from one cell side or
the other, usually differ---which renders the global solution discontinuous
in space.

%
%
The following recapitulation of ADER-DG's algorithmic steps
neglects sources and non-con\-servative terms for
simplicity.
We start from a weak formulation of (\ref{equation:introduction:PDE}) and 
discretise it in space.
Let our weak formulation be a Ritz-Galerkin formulation, i.e.~let the test
and shape space hold the same functions multiplied by $p$th order
time-dependent polynomials.
Partial integration over the time interval $(T,T+\Delta T)$ and integration
over the cells $c$ by applying Green's theorem then yield an implicit formula
for the update $D_h = Q_h(T+\Delta T) - Q_h(T)$.
ADER-DG's fundamental idea is to replace $Q_h$ and, thus, $F(Q_h)$ with an
``estimated'' evolution $Q_h^*$ within $(T,T+\Delta T)$ plus an 
additional term.
$Q_h^*$ is a predicted solution in space and time---the {\em space-time
predictor}---that results from a set of implicit solves of
(\ref{equation:introduction:PDE}) restricted to individual cells.
We obtain the discrete variational problem \cite{Dumbser:14:Posteriori}

\begin{align}
\label{equation:aderdg:update}
\int_c
v_h\,
D_h
\text{d}x
&=
\int_T^{T+\Delta T}
\int_c
\nabla v_h
:
F(Q_h^*)
\,
\text{d}x
\text{d}t
-
\int_T^{T+\Delta T}
\oint_{\partial c}
v_h [F(Q_h^*) \cdot n] _{\partial c} \,
\text{d}s
\text{d}t.
\end{align}

\noindent
Three algorithmic steps that decompose into further tasks numerically invert
\eqref{equation:aderdg:update}.

\paragraph{ADER-DG step 1: compute space-time predictor}

To obtain a prediction $Q_h^*$, we study the individual cells $c$ independent of
each other, i.e.~we neglect that the solutions of neighbouring cells
interfere.
This gives

\begin{equation}
  \mbox{\texttt{STP:}}
  \qquad
  \forall c \in \mathbb{T}:
  \quad
  \int _{T}^{T+\Delta T}
  \int _c
  \left(
   \frac{\partial Q_h^*}{\partial t}
   +
   \nabla\cdot F(Q_h^*)
   \right)
   \, \varphi_h
   \,
  dxdt
  =0.
  \label{equation:aderdg:stp}
\end{equation}

\noindent
Cell-wisely, it is an implicit scheme that we solve through Picard
iterations.
Lagrange polynomials with Gauss-Legendre collocation points plus Green's theorem
yield a
quadrature-free formulation with sparse operators.
Yet, the number of required Picard iterations per cell is not known
and the solve (of one step) of (\ref{equation:aderdg:stp}) remains
arithmetically expensive.

\texttt{STP} works with the PDE solution in
space and time, i.e.~a $d+1$-dimensional dataset.
Technically, it decomposes into a task \texttt{predict} that spans the
polynomial in space and time, and a task \texttt{extrapolate} that yields the
predicted solution plus its flux along the normal on the cell faces.
The extrapolation is needed by the next step.

\paragraph{ADER-DG step 2: solve Riemann problem}

The projection of both the solution $Q_h^*$ 
and its normal flux $F(Q_h^*) \cdot n$ of (\ref{equation:aderdg:stp}) 
onto the cell faces exhibits jumps.
$n$ is the generic normal vector spanning a row-wise
scalar product.
$Q_h^{*\pm }$ and $F(Q_h^*)^{\pm} \cdot n$
from one side ($+$) and the other ($-$) typically
differ; not only by a sign.
ADER-DG's second step thus replaces the discontinuous expression by a
numerical flux.
It applies a Rusanov solver.
This yields
\begin{equation}
  \mbox{\texttt{Riemann:}}
  \quad
  \forall c \in \mathbb{T}:
  \quad
  [F(Q_h^*) \cdot n]  _{\partial c} =
  F^*_h \ \mbox{with} \
  F^*_h = F^*_h(Q_h^{*\pm }, F(Q_h^*)^{\pm }n).
  \label{equation:aderdg:Riemann}
\end{equation}

\noindent
The left and right predicted
solution onto a face are required as input, as well as the left and right
predicted solution flux along the normal.
\texttt{Riemann} yields flux contributions over $(T,T+\Delta T)$ along the
cell faces.
We formalise it as a task
\texttt{solveRiemann} determining $F^*_h$, which is formally a $F^*_h(t)$
with $t \in (T,T+\Delta T)$.
We note that holding $F_h^{*\pm}$ is redundant for 
$B=0$ as the two sides differ only by their sign.
Yet, it allows each cell to work with ``its'' Riemann solve result and
applies without modifications to $B \not = 0$.

\paragraph{ADER-DG step 3: correct predicted solution}

We finally return to the non-local
formulation (\ref{equation:aderdg:update}) which is made
globally explicit through $Q_h^*$ and $F_h^{*\pm}$.
This explicit forward operator evaluation renders ADER-DG a single-step scheme
and decomposes (\ref{equation:aderdg:update}) into a cell-local contribution which we have
available from ADER-DG's first step plus a contribution from the Riemann
problem.
We add the Riemann solve's result to the predicted value.
It is corrected.
Hence, our third step
is denoted as \texttt{Corrector}.

The \texttt{Corrector} step decomposes into a task
\texttt{integrateVolume} collapsing the space-time polynomial of the
predictor onto $T+\Delta T$,
an \texttt{integrateFace} task integrating over the result of the Riemann solve 
and thus projecting the space-time evolution along the cell faces onto 
the cells' spatial data at $T+\Delta T$, an 
\texttt{update} task which swaps $D_h$ into the location of $Q_h$, and \texttt{calculateTimeStep} which
returns an admissible time step size $\Delta T_{adm}$ for the next time step
due to the CFL condition.
The overall predictor-corrector sequence equals

\begin{eqnarray}
  Q_h (T+\Delta T) & =  &
  \left(
  id+
  \texttt{integrateVolume}
  \circ
  \texttt{predict}
  +
  \texttt{integrateFace}\ 
  \circ
  \right.
  \nonumber \\
  &&\
  \left.
  \texttt{solveRiemann}
  \circ
  \texttt{extrapolate}
  \circ
  \texttt{predict}
  \right)
  Q_h (T)
  \label{equation:operator-multiple-evaluations}
  \\
  & = &
  \left(
  id+
   \left(
  \texttt{integrateVolume}
  +
  \texttt{integrateFace}
    \ \circ \right. \right.
  \nonumber \\
  &&\
  \left. \left.
  \texttt{solveRiemann}
  \circ
  \texttt{extrapolate}
  \right)
  \circ
  \texttt{predict}
  \right)
  Q_h (T)
  \label{equation:operator}
\end{eqnarray}

\noindent
as operator evaluation with its memory footprints displayed in Table
\ref{table:memory-requirements}.

\addtolength{\tabcolsep}{-4pt}   
\begin{table}
 \caption{
  Memory requirements per task as number of double values (degrees
  of freedom): how much data per task is read in and how much data is written.
  All data is normalised per task and per cell $c$ or neighbour cell pair
  $c_a$,$c_b$.
  We group the tasks along three steps \texttt{STP}, \texttt{Riemann} and
  \texttt{Corrector}.
  \label{table:memory-requirements}
 }
 \centering
 \footnotesize 
 \begin{tabular}{l|rrp{3.8cm}} 
 Task & \multicolumn{1}{c}{in} & \multicolumn{1}{c}{out}  & Remarks \\
 \hline
 \hline
 \texttt{STP} \\
 \texttt{predict}($c$) & $m \cdot (p+1)^{d\phantom{+1}}$ & $(d+1) \cdot m \cdot
 (p+1)^{d+1}$ & determine $Q_h^*$ and $F(Q_h^*)$ \\
 \texttt{extrapolate}($c$) & $(d+1) \cdot m \cdot (p+1)^{d+1}$ & $4d\cdot m
 \cdot (p+1)^{d\phantom{+1}}$ & extrapolate $Q_h^*$ and $F(Q_h^*)\cdot n$ to all
 $2d$ faces of a cell \\
 \hline
 \hline
 \texttt{Riemann} \\
 \texttt{solveRiemann}($c_a$,$c_b$) & $4d \cdot m \cdot (p+1)^{d\phantom{+1}}$ &
 $2d\cdot m \cdot (p+1)^{d\phantom{+1}}$
  & may store result in one of the input arrays
  \\
 \hline
 \hline
 \texttt{Corrector} \\
 \texttt{integrateVolume($c$)}
  & $d\cdot m \cdot (p+1)^{d+1}$ & $m \cdot (p+1)^{d\phantom{+1}}$ 
  & integrate predicted $F(Q_h^*)$
  \\
 \texttt{integrateFace}($c$) & $2d\cdot m \cdot (p+1)^{d\phantom{+1}}$ & $m
 \cdot (p+1)^{d\phantom{+1}}$ & accumulate in output data structure of
 \texttt{integrateVolume} \\
 \texttt{update}($c$)  & $m \cdot (p+1)^{d\phantom{+1}}$  & $m \cdot
 (p+1)^{d\phantom{+1}}$ & add $D_h$ to $Q_h$
 \\
 \texttt{calcTimeStep}($c$) & $m \cdot (p+1)^{d\phantom{+1}}$ & 1
 $\phantom{\cdot (p+1)^{d+1}}$
 \end{tabular}
\end{table}
\addtolength{\tabcolsep}{4pt}   

\paragraph{Time stepping}

If the combination of PDE plus initial and boundary conditions yields
a sufficiently smooth solution, the convergence of ADER-DG
is numerically verified to converge with
optimal order; as long as $\Delta T$ does not harm the CFL condition
\cite{Dumbser:06:ADERDG,Gassner:11:ExplicitDG}.
Two requirements arise for our implementation from this statement:
On the one hand, a robust implementation has to determine 
$\Delta T_{adm}$ numerically, and it has to ensure these are not exceeded.
On the other hand, 
a $p$th order scheme is known to yield oscillating
solutions and non-physical values (such as negative densities) once 
discontinuous solutions (``shocks'') arise \cite{Hesthaven:08:NodalDG}.
A robust time stepping thus has to be able to identify
localised nonphysical solutions ({\em troubled cells}) after the
\texttt{Corrector}, roll back to the old time step and recompute the solution
through Finite Volumes. 
The Finite Volumes act as limiter.


Global time stepping has to restrict $\Delta T_{adm}$ once all
\texttt{Corrector} tasks have terminated, before it broadcasts the next
time step size to use to all ranks.
Local time stepping, i.e.~time stepping where the time step size
depends with $\Delta T_{adm}/3^{\ell}$ on the mesh size ($\ell $ is the
spacetree level) equals global time stepping w.r.t.~the exchange of $\Delta T_{adm}$.
Yet, the restriction can be temporarily sparsified, i.e.~is to be done for the
coarsest cells only.
Local time stepping in
ADER-DG advances cells on the levels $\ell$ and $\ell +1$ through \texttt{STP}.
As $\Delta T$ depends on $\ell$, the finer cells on level $\ell+1$ now can
evaluate ``their'' \texttt{Riemann} problem and \texttt{Corrector} step.
This Riemann solve tackles the jump adjacent to the coarser cell
partially in time.
It thus yields a partial update of the coarse neighbour's next time step's
solution through \texttt{integrateFace}.
Since we rely on tripartitioning, the coarser cells on level $\ell $ cannot run the Riemann solve until
the adjacent finer cells on level $\ell +1$ have made three time steps.
This pattern extends to non-balanced grids recursively,
yet  assumes that $\Delta T$ depends smoothly on the mesh size.
If this does not hold, more than three steps per finer grid level are
required.
Global $\Delta T$ coordination in this case becomes necessary.
As a result, we restrict our analysis to global time stepping.
Local time
stepping realises global data flow and data exchange along the same lines though with
fewer reductions.

%
%
\paragraph{Finite Volume limiter}
Following \cite{Dumbser:14:Posteriori,Zanotti::15::SpaceTimeLimitingADERDG}, 
we augment the ADER-DG scheme with adaptive mesh refinement plus subcell
limiting.
Let {\em troubled} cells be cells where the hosted ADER-DG solution
is inadmissible.
It oscillates or does physically not make any sense for example.
Besides the ADER-DG data, we make these cells host an additional regular
Cartesian subgrid of the resolution $(2p+1)^d$ and advance a Finite
Volumes (FV) scheme on this subgrid forward in time.
After each time step, averaging projects the FV solution back onto the DG space
and then assesses whether the FV solution can be represented by the polynomials
without harming the admissibility.
Once this is the case for all $m$ components, we remove
the Cartesian FV grid, remove the troubled marker and continue
with the ADER-DG scheme within the cell.

\begin{corollary}
 \label{corollary:FV}
 Finite Volume schemes with Rusanov fluxes or a MUSCL-Hancock Riemann solver
 neglecting cell interactions that are not face-connected can be cast into our
 predictor-corrector formalism.
\end{corollary}

\noindent
The proofs (Appendix \ref{appendix:aderdg}) are technical but clarify that
it is sufficient to study plain ADER-DG's memory
movement and communication-behaviour only.
If we couple ADER-DG with FV, both solvers exhibit the
same communication pattern.
All communication characteristics change only quantitatively.

\section{Statements on ADER-DG}
\label{section:theorems}

Let the term {\em element-wisely} characterise a grid-based algorithm where an
operator on a cell has no access to any other cell.
It may solely access a cell's data and its faces.
Element-wise operators working on faces may access only face data.
Let the term {\em single-touch} characterise an implementation which reads and 
writes unknowns only once per time step \cite{Weinzierl:17:BoxMG}. 

\begin{constraint}
 Multiple \texttt{STP} evaluations have to be avoided.
\end{constraint}

\noindent
We anticipate the experimental insight that the
prediction in variation (\ref{equation:operator}) dominates the computational
cost. 
Though (\ref{equation:operator-multiple-evaluations}) and
(\ref{equation:operator}) are symbolically the same, 
only (\ref{equation:operator}) thus is computationally feasible.

\begin{constraint}
 We focus on problems where (i) one solution snapshot $Q_h$ does not fit into the registers and machine
 caches, and (ii) one cell's space-time polynomial, i.e.~all data read and written by
 \texttt{predict}, fit into the registers/caches.
\end{constraint}

\noindent
Any more relaxed assumption on problem sizes or machine configurations 
relaxes our single-touch statements.
We abstract from different cache levels and their multifaceted
impact on real-world performance.
For most machines, our memory transfer statements apply at least to
the last-level cache.


\begin{theorem}[Memory footprint theorem]
 \label{theorem:memory}
 Even though ADER-DG's prediction spans a space-time polynomial, it is
 sufficient for an element-wise ADER-DG realisation to store a
 space-time hull of each space-time hypercube in the mesh persistently.
 This hull has to encode flux contributions along faces besides solution data.
\end{theorem}

\begin{theorem}[Weak element-wise single-touch theorem]
 \label{theorem:weak-single-touch}
 There is a $1 \leq C_\textnormal{rerun} <2$ such that the predictor-corrector operator in
 the formulation (\ref{equation:operator}) can be realised element-wisely with
 $C_\textnormal{rerun}$ data reads per time step.
\end{theorem}

\noindent
Both proofs are constructive and apply the techniques from
Sec.~\ref{section:tasks}.
An experimental $C_\textnormal{rerun} \approx 1$ renders the single-touch
property strong. 
If we skipped the term element-wisely, the theorem would become trivial
as (\ref{equation:operator}) would be just one global explicit operator
consisting of two components.

\begin{corollary}
 \label{corollary:naive-element-wise-single-touch}
 A direct element-wise realisation of (\ref{equation:operator}) can not be
 single-touch. 
\end{corollary}

\noindent
This proof is technical (Appendix \ref{appendix:theorems}) and
relies on the fact that the steps are
subject to a partial temporal order (Fig.~\ref{figure:aderdg:eyecatcher})

\begin{eqnarray}
 \forall c_a,c_b \in \mathbb{T}: \texttt{STP}(c_a)
  & \before &
  \texttt{Riemann}(c_a,c_b)
  \qquad
  \wedge
  \nonumber
  \\
 \texttt{Riemann}(c_a,c_b)
  & \before &
  \texttt{Corrector}(c_a).
  \label{equation:tasking:dependencies}
\end{eqnarray}

\noindent
While it motivates some of
the techniques we use later to construct a single-touch element-wise algorithm,
it demonstrates that applying an isolated, single technique alone is not
sufficient.

\section{Communication-avoiding, task-based ADER-DG}
\label{section:tasks}

ADER-DG tasks are subject to the temporal constraints  
(\ref{equation:tasking:dependencies}).
Additional tasks may enter the task graph:
\begin{enumerate}[leftmargin=*]
  \item If requested by the user, we insert a \texttt{Plot} task after the
    \texttt{Corrector}.
  \item A time step size computation follows the solution update.
  \item If inadmissible high order solutions are to be expected,
    we merge an
  \texttt{Admissible} task into \texttt{Corrector} which can label cells as
    {\em troubled}. If a cell becomes troubled, we roll back to the previous
    solution, insert the FV's patch representation of this solution, and 
    add an additional $\texttt{STP}$ task rerunning the time step with FV.
  \item If a cell is subject to FV, we add an additional
    \texttt{Reconstruct} step after the corrector that projects the Finite
    Volume solution back onto the higher-order ADER-DG representation. If a FV
    solution could validly be represented by the polynomial, the cell's state is
    reset and the Cartesian patch data are erased. 
  \item To couple ADER-DG solutions with FV, there is a helper
    task type inserting ``non-troubled'' FV cells around
    the troubled ones. Such a halo solely represents the DG solution as
    Finite Volumes, yet does not run FV updates. The cells allow the FV scheme
    to compute fluxes on the patch faces of troubled cells. They act as coupling layer between ADER-DG and the limiter.
  \item If adaptive mesh refinement (AMR) is used, additional AMR tasks realising
    the refinement and coarsening criteria are to be inserted.
  \item Additional tasks for boundary and initial conditions---both overwrite
    the computed solution with a prescribed or altered data---complete the
    picture.
\end{enumerate}

\begin{algorithm}[htb]
 \caption{
  One time step of the simplest version of ADER-DG without limiter as
  traditional, sequential pseudo code.
  Left: Standard time stepping.
  Right: Shifted time stepping.
  $C_{\Delta T} \leq 1$ is a chosen security factor.
  \label{algorithm:time-stepping}
 }
 \begin{minipage}{0.49\textwidth}
  \begin{algorithmic}[1]
   \While{$T<T_{final}$}
  \For{$c \in \mathbb{T}$}
      \Comment \texttt{STP}
    \State $\texttt{predict}(c)$
      \Comment use $\Delta T$
    \State $\texttt{extrapolate}(c)$
  \EndFor
  \For{face-connected $c_a,c_b \in \mathbb{T}$} 
    \State  \Comment \texttt{Riemann}
    \State $\texttt{solveRiemann}(c_a,c_b)$
  \EndFor
  \State $\Delta T_{adm} \gets \infty$
  \For{$c \in \mathbb{T}$}
    \Comment \texttt{Corrector}
    \State $\texttt{integrateVolume}(c)$
    \For{all faces $f$ of $c$}
      \State $\texttt{integrateFace}(f)$
    \EndFor
    \State $\texttt{update}(c,f_1,f_2,\ldots)$
    \State $\Delta T_{adm} \gets \min\{ $
    \State \phantom{xx} $\Delta T_{adm}, \texttt{calcTimeStep}(c)\}$
  \EndFor
  \State $T \gets T + \Delta T$ 
  \State $\Delta T \gets C_{\Delta T} \cdot \Delta T_{adm}$
\EndWhile
 
  \end{algorithmic}
 \end{minipage}
 \begin{minipage}{0.49\textwidth}
  \begin{algorithmic}[1]
   \While{$T<T_{final}$}
  \For{face-connected $c_a,c_b \in \mathbb{T}$} 
    \State  \Comment \texttt{Riemann}
    \State $\texttt{solveRiemann}(c_a,c_b)$
  \EndFor
  \State $\Delta T_{adm} \gets \infty$
  \For{$c \in \mathbb{T}$}
    \Comment \texttt{Corrector}
    \State $\texttt{integrateVolume}(c)$
    \For{all faces $f$ of $c$}
      \State $\texttt{integrateFace}(f)$
    \EndFor
    \State $\texttt{update}(c,f_1,f_2,\ldots)$
    \State $\Delta T_{adm} \gets \min\{ $
    \State \phantom{xx} $\Delta T_{adm}, \texttt{calcTimeStep}(c)\}$
  \EndFor
  \State $T \gets T + \Delta T$ 
  \State $\Delta T \gets C_{\Delta T} \cdot \Delta T_{adm}$
  \For{$c \in \mathbb{T}$}
    \Comment \texttt{STP}
    \State $\texttt{predict}(c)$
    \State $\texttt{extrapolate}(c)$
  \EndFor
\EndWhile
 
  \end{algorithmic}
 \end{minipage}
\end{algorithm}

%
%
\noindent
Our ADER-DG scheme is defined over a computational grid.
Given the grid plus the task dependencies
stemming from (\ref{equation:tasking:dependencies}), we obtain a global task
graph per time step.
We interpret the grid instantiating the task graph at hands of
(\ref{equation:tasking:dependencies}) \cite{Weinzierl:17:Peano}.
As ADER-DG's task dependencies are localised, we can run over the grid (in
parallel) and issue the tasks directly per cell.
There is no need to assemble the task graph.

If we translate ADER-DG's steps straightforwardly into a traversal code, we
obtain a sequence of loops (Alg.~\ref{algorithm:time-stepping}). 
The loop body per cell or face, respectively, is the actual task invocation.
We refer to such an implementation as straightforward.
One loop issuing \texttt{STP}
with high arithmetic intensity is followed by two loops
triggering computationally cheap tasks.
Yet, we observe that, on the one hand, ADER-DG imposes only a partial order,
and we observe that not all tasks have to be tied to ``their'' step.
We can reorder and rearrange the execution order as long as 
all data remains consistent.
On the other hand, ADER-DG does not
enforce us to run the same operation on all grid entities.
It allows us to run the predictor in one grid area, while another area
performs Riemann solves.
There is no need to align each grid traversal with exactly one
task type.
These observations give rise to the following techniques.
Throughout their presentation, we make a {\em realisation time step}
comprise all three task types \texttt{STP}, \texttt{Riemann} and
\texttt{Corrector}.

%
%
\begin{technique}
  \label{technique:integrate-early}
  Move tasks from one logical step into an earlier step such that they are
  ran as early as possible.
\end{technique}

\noindent
This technique is data-centric as it helps to avoid capacity and conflict
cache misses as well as register spilling.

\begin{implication}
 We evaluate \texttt{integrateVolume} directly after \texttt{predict} and 
 evaluate \texttt{integrateFace} directly once \texttt{solveRiemann}
 has terminated. This allows us to eliminate temporary variables.
\end{implication}

\noindent
\texttt{predict} spans the space-time polynomial.
Applying the technique, this outcome is to be held only temporarily:
It is passed over to \texttt{extrapolate} before we integrate it and discard the
data.
The integration is a task logically assigned to the \texttt{Corrector} step, 
i.e.~we bring it forward.
Its result is stored in $D_h$ instead of the old time step's $Q_h$ as we have to
preserve $Q_h$ to facilitate rollbacks required for a FV intervention.

The technique is applied analogously to \texttt{integrateFace}.
\texttt{solveRiemann} handles jumps over a face in space-time.
\texttt{integrateFace} maps this data of spatial cardinality $m(p+1)^{d-1}$ plus
a temporal dimension onto a contribution to the solution update $D_h$.
Different to the space-time predictor, moving the integration from the
\texttt{Corrector} into \texttt{Riemann} does not allow us to reduce the memory
footprint.
Adaptive time stepping requires us to hold $Q_h^{*\pm}$ and $F(Q_h^{*\pm}) \cdot
n$ over the whole space-time span of a cell to determine (partial) $F_h^{*\pm}
\cdot n$.
Passing the two $F_h^{*\pm} \cdot n$ ingredients directly into
\texttt{integrateFace} would violate our definition of cell-wisely.


%
%
\begin{technique}
  \label{technique:different-tasks}
  Run different tasks on different grid entities concurrently.
\end{technique}

\noindent
Graph-based tasking systems model tasks as nodes in a graph and the task
dependencies as edges.
At one particular time, a tasking system may launch any task which has no
pending dependencies anymore.
Once a task terminates, the node plus all outgoing edges (data writes) are
removed from the graph.
The power of tasking \cite{Reinders:07:TBB,YarKhan:17:PlasmaOpenMP} results from
the fact that different task types might be launched simultaneously.

This idea translates to grid-based solvers:
When we run through the grid, we impose a spatial order on grid entities.
We may launch the Riemann solve whenever we run into a face for the first time,
i.e.~when we read it from the main memory the first time (touch-first
semantics).
If we stick to a pure cell-wise traversal, the time when a face is ``touched''
for the very first time results directly from the cell ordering:
if we enter a cell, we touch all of its faces; maybe some of them for the first
time.

\begin{implication}
 When we enter a cell to perform \texttt{Corrector}, we analyse all of its $2d$
 adjacent faces whether they have been read from main memory before.
 For faces that have not yet been used, we run the \texttt{Riemann} solve. 
 As the \texttt{Corrector} issues the tasks from the \texttt{Riemann} step
 on-demand, \texttt{Corrector} and \texttt{Riemann} evaluations are intermixed
 throughout one grid sweep.
 We eliminate an explicit \texttt{Riemann} sweep.
\end{implication}

\noindent
This technique is particularly interesting in the context of
cache-oblivious algorithms \cite{Kowarschik:03:CacheOverview}.
Let a grid traversal minimise the time between when one cell is handled to
the handling of any neighbour cell.
The first cell's \texttt{Corrector} requires the outcome of the Riemann solves
of its faces. If not computed yet, it determines them.
The neighbouring cell uses these face data, too.
If the time in-between the two task executions is small, we have 
temporarily localised data access:
the probability that the data still resides in cache is high.
Traditional techniques to obtain such traversal orderings are loop blocking and
space-filling curves \cite{bader_space_filling_2013}.
In our experiments, we rely on the latter \cite{Weinzierl:11:Peano}.

%
%
\begin{technique}
  \label{technique:shift}
  Shift the tasks by half a grid sweep.
\end{technique}

\noindent
We kick off the computation with an evaluation of
\texttt{Riemann} and \texttt{Corrector}, before we run the first complete triad
\texttt{STP},\texttt{Riemann},\texttt{Corrector}.
Furthermore, we distinguish time step sizes $\Delta
T_{adm}$, $\Delta T_{old}$ and $\Delta T_{new}$. 
\texttt{Riemann} and \texttt{Corrector} use $\Delta T_{old}$.
\texttt{STP} uses $\Delta T_{new}$.
After \texttt{STP}, we roll over $\Delta T_{old} \gets \Delta T_{new}$.
Formally, we recast the time stepping in (\ref{equation:operator}) into a
sequence of

{\footnotesize
\[
  \left(
  \begin{array}{c}
    Q_h \\
    D_h \\
    (Q_h^{*\pm },F(Q_h^*)^\pm \cdot n )
  \end{array}
  \right)
  (T+\Delta T)
  = 
  \underbrace{
    \left( 
      \begin{array}{ccc}
       1 & 0 & 0 \\
       \texttt{integrateVolume} \circ \texttt{predict} & 0 & 0 \\
       \texttt{extrapolate} \circ \texttt{predict} & 0 & 0 \\
      \end{array}
    \right)
  }_{
  \mbox{run \texttt{STP} of next time step already}
  }  
  \cdot
\]

\begin{equation}
  \underbrace{
    \left( 
      \begin{array}{ccc}
        0 & 1 & 
        \texttt{integrateFace}
        \circ
        \texttt{solveRiemann} 
        \\ 0 & 0 & 0 
        \\ 0 & 0 & 0 
      \end{array}
    \right)
  }_{
  \mbox{complete \texttt{Corrector} and \texttt{Riemann} started in previous
  sweep} }
  \left(
  \begin{array}{c}
    Q_h \\
    D_h \\
    (Q_h^{*\pm },F(Q_h^*)^\pm \cdot n )
  \end{array}
  \right)
  (T)
  \label{equation:tasking:shift}
\end{equation}
}

\noindent
applications (Alg.~\ref{algorithm:time-stepping}).
Here, \texttt{integrateVolume} is already merged into \texttt{STP}
(cmp.~Technique \ref{technique:integrate-early}).
While the operator matrices enlist \texttt{predict} twice, it has to
be evaluated only once with the result temporarily held and passed into two
follow-up tasks.
As $\Delta T_{old}=0$ throughout the first grid sweep, the very first
composition of \texttt{Riemann} and \texttt{Corrector} does not modify the solution as
the integrals in (\ref{equation:aderdg:stp}) and (\ref{equation:aderdg:Riemann})
degenerate.

\begin{implication}
Let $3N$ grid sweeps through the grid be used by a straightforward ADER-DG
realisation.
Each sweep implements one ADER-DG step (predictor, Riemann, corrector).
The
Techniques \ref{technique:different-tasks} and \ref{technique:shift} make the
solve require $N+1$ sweeps in total.
\end{implication}

\noindent
Formalism (\ref{equation:tasking:shift}) logically shifts the evaluation of 
ADER-DG by half a grid sweep:
Half of the work of a subsequent time step is already done once a sweep
terminates.
We thus refer to this realisation variant as {\em shifted
evaluation}.
The approach resembles pipelining
\cite{Ghysels:13:HideLatency,Reps:15:Helmholtz}.
Different to most pipelining, our shifts do however not
introduce additional helper variables.
$D_h$ and $(Q_h^{*\pm },F(Q_h^*)^\pm \cdot n )$ have to be
held in memory anyway to facilitate roll backs for the limiter and local time
stepping.

\begin{implication}
Let $\Delta T_{adm}>T_{new}$.
Technique \ref{technique:shift} in combination with Technique
\ref{technique:different-tasks} allows to realise, amortised, one realisation
time step per grid sweep.
\end{implication}

\noindent
The implication yields $1 \leq C_\textnormal{rerun}$ in Theorem
\ref{theorem:weak-single-touch}.

\begin{technique}
  \label{technique:fuse}
  Fuse tasks.
\end{technique}

\noindent
Our shifted grid sweeps run through the cells of $\mathbb{T}$.
Per cell, they analyse the $2d$ faces and trigger \texttt{Riemann} if the
neighbouring cell connected through the face has not been traversed before.
As soon as this cell preamble has terminated, 
the grid sweep triggers \texttt{Correction} and immediately afterwards \texttt{STP}.
All orders obey (\ref{equation:tasking:dependencies}).
Algorithmically, this approach is similar to loop fusion.
Some faces might still ``wait'' for a Riemann solve, while others already have
finished their correction and projected $Q_h^*$ and $F(Q_h^*)$ onto the
faces again.
From a cell's point of view, we invert and fuse the task execution order:
Per cell, \texttt{STP} is immediately launched after \texttt{Corrector}.
This is made possible as we distinguish $\Delta T_{new}$ and  $\Delta T_{old}$.
We fuse \texttt{Corrector} and \texttt{STP} (in this order).
Also the concatenation of \texttt{integrateVolume}
with \texttt{predict} is logically task fusion.

\begin{implication}
Techniques  \ref{technique:different-tasks}, \ref{technique:shift}  and
\ref{technique:fuse} read, amortised, each cell's solution data $Q$
only once per time step as long as $\Delta T_{adm}>T_{new}$.
\end{implication}

\noindent
Even if a code does not physically fuse tasks, any
%
%
%
output data of the first task remains in the cache as input to the
second---subject to a sufficiently big cache---as we run the respective tasks
directly after each other.
Since we schedule our tasks ourselves through the grid traversal, no task
stealing or thread switching can interfere.

%
%
\begin{technique}
Be optimistic.
\end{technique}

\noindent
The explicit nature of ADER-DG requires us to exchange $\Delta T_{adm}$
once a time step completes.
This introduces two synchronisation points in our shifted implementation: 
the grid sweep itself plus the time step size synchronisation.
%
%
Two optimistic modifications allow us to eliminate one synchronisation
without loosing adaptive time step size choices.
On the one hand, we run \texttt{STP} with $\Delta T_{new}$ ignorant of  $\Delta
T_{adm}$: 
As soon as a corrector finishes, its cell is subject to the next predictor
though other correction steps running concurrently or later still might reduce
$\Delta T_{new}$.
On the other hand, we check if $\Delta T_{adm}<\Delta T_{new}$ after all
\texttt{STP}s of one time step have terminated. 
If we find that our chosen time step size harms the CFL 
condition, we reset $\Delta T_{old} \gets C_{\Delta T} \cdot \Delta T_{adm}$ and
rerun the prediction.

A naive implementation would break out of the cell loop immediately once the CFL
condition is harmed.
We however note that this can be disadvantageous as we might run into rippling
loop restarts.
Running all \texttt{STP}s always, thus finishing all updates and $\Delta
T_{adm}$ computations, and reducing $\Delta
T_{adm}$ afterwards implies that the set of \texttt{STP}s is relaunched either
not at all or exactly once.
This is {\em optimistic time stepping} as all \texttt{Corrector}-\texttt{STP}
task pairs are ran even though we might already know after the first few that
the outcome of the second task per pair is invalidated.

\begin{implication}
Even if the optimistic assumption fails, our code requires at most two grid
sweeps per realisation time step.
As we may assume that not every time step harms the CFL condition, we obtain a
real upper bound $C_\textnormal{rerun} < 2$.
\end{implication}

\noindent
In our implementation, we make $\Delta T_{new} = \frac{1}{2}\left( T_{old} +
C_{\Delta T} \cdot \Delta T_{adm} \right)$.
As long as time step sizes are chosen too pessimistic, we make $\Delta T_{new}$
creep towards $\Delta T_{adm}$.

\section{Properties}
\label{section:properties}

\begin{algorithm}[htb]
 \caption{
  Our cell-based, shifted ADER-DG implementation relying on all
  introduced optimisation techniques.
  \label{algorithm:algorithm-efficient-ADER-DG}
 }
 \begin{algorithmic}[1]
  \While{$T \leq T_{final}$}
 \If{$\Delta T_{adm}<\Delta T_{new}$}
    \Comment Optimistic guess had been wrong
  \State $(\Delta T_{old},\Delta T_{new})\gets C_{\Delta T} \cdot \Delta
  T_{adm}$ \Comment Reset time step sizes
  \For{$c \in \mathbb{T}$}
    \Comment \texttt{STP} rerun is additional loop over cells
    \State $\texttt{predict}(c)$
    \State $\texttt{extrapolate}(c)$
    \State $\texttt{integrateVolume}(c)$
  \EndFor
 \EndIf
 \State $\Delta T_{old} \gets \Delta T_{new}$
 \State $\Delta T_{new} \gets C_{\Delta T} \cdot \Delta T_{adm}$
  or $\Delta T_{new} \gets 0.5(\Delta T_{old}+C_{\Delta T} \cdot \Delta T_{adm})$
 \State \Call{timeStep}{$\Delta T_{old},\Delta T_{new}$}
  \Comment{Either strict time step set or creeping average}
\EndWhile
\Function{timeStep}{$\Delta T_{old},\Delta T_{new}$}
 \State $\Delta T_{adm} \gets \infty$
 \For{$c \in \mathbb{T}$}
   \Comment{Element-wise traversal}
  \For{$f \in \mbox{adjacent faces of }c$}
    \Comment \texttt{Riemann}
   \If{$f$ touched first time}
    \State $\texttt{solveRiemann}(f^{+},f^{-})$
   \EndIf
  \EndFor
  \For{$f \in \mbox{adjacent faces of }c$}
    \Comment \texttt{Corrector} using $\Delta T_{old}$
   \State $\texttt{integrateFace}(f)$
  \EndFor
  \State $\texttt{update}(c,f_1,f_2,\ldots)$
    \Comment{$f_1,f_2,\ldots$ are adjacent faces of $c$}  
  \State $\Delta T_{adm} \gets \min\{ \Delta T_{adm}, \texttt{calcTimeStep}(c)\}$
  \State $\texttt{predict}(c)$
    \Comment Immediately kick off next \texttt{STP}
  \State $\texttt{extrapolate}(c)$
    \Comment subject to $\Delta T_{new}$
  \State $\texttt{integrateVolume}(c)$
    \Comment anticipate \texttt{Corrector} task, hold only space-time hull
 \EndFor
\EndFunction

 \end{algorithmic}
\end{algorithm}

With our four communication-avoiding techniques defined, we end up with a
cell-wise algorithm (Alg.~\ref{algorithm:algorithm-efficient-ADER-DG}).
Several properties arise directly.

\begin{corollary}
 \label{corollary:hull}
The persistent memory footprint of our algorithmic realisation variant is 
\[
  \frac{2+6d}{(d+1)(p+1)+1+6d} \leq 0.875
\]
of the footprint of a straightforward element-wise ADER-DG implementation
holding the output of the three ADER-DG steps explicitly.
\end{corollary}

\begin{proof}
 A three-step, straightforward element-wise ADER-DG implementation
 \linebreak starts from $Q_h$ to span the space-time polynomial.
 An \texttt{STP} epilogue derives explicitly the extrapolated predicted solution
 plus its flux contribution along the faces.
 The Riemann solver determines numerical fluxes which are stored on the face
 over the whole time span of the adjacent cells.
 For local time stepping, the adjacent cell with the smaller span determines
 ``whole time span''.

 Technique \ref{technique:integrate-early} implies that we hold the space-time representation
 only temporarily as the outcome of \texttt{integrateVolume} is backed up in $D_h$.
 Also \texttt{integrateFace} directly accumulates all outputs into this array.
 $D_h$ has to be held separately from $Q_h$ as we have to give ADER-DG the
 opportunity to declare a time step as troubled and to fall back to FV.
 As \texttt{integrateVolume} and \texttt{integrateFace} are brought
 forward in simulation time, it is not possible to merge admissibility checks
 into the $D_h$ update, i.e.~$D_h$ cannot be eliminated.

 We combine these insights with the analysis from Table
 \ref{table:memory-requirements}.
\end{proof}

\noindent
%
It is the predictor's full space-time representation that we do not have to
hold.
We hold solely the space-time hull of each cell with additional flux data on
the space-time faces.
The memory pressure compared to a straightforward version reduces as 
the memory volume is shrinked.
Our approach falls into the first rubric of communication-avoiding techniques.

\begin{corollary}
\label{corollary:rerun-constraint}
Let $T_\textnormal{3\ steps}$ be the time required to run the triad of \texttt{STP},
\texttt{Riemann} and \texttt{Corrector} as three separate grid sweeps. Let
$T_\textnormal{STP}<T_\textnormal{3\ steps}$ be the time that is spent on a sole
\texttt{STP} sweep.
$T_\textnormal{fused}$ is the time
required to run one fused time step without any reruns.
Let the used $\Delta T = C_{\Delta T} \cdot \Delta T_\textnormal{adm}$ be damped
by $C_{\Delta T} \leq 1$.
Our communication-avoiding fused algorithm yields better time-to-solution than
a straightforward approach if
\[
1 \leq
C_\textnormal{rerun} <
1+
\frac{ T_\textnormal{3\ steps} \cdot C_{\Delta T} - T_\textnormal{fused} }
     { T_\textnormal{STP} }.
\]
\end{corollary}

\begin{proof}
Our weak single-touch Theorem \ref{theorem:weak-single-touch} in combination
with Algorithm \ref{algorithm:algorithm-efficient-ADER-DG} clarifies that the
total runtime of our fused approach, including any reruns, equals \linebreak
$( (C_\textnormal{rerun}-1)T_\textnormal{STP} + T_\textnormal{fused} ) \cdot
N_\textnormal{fused}$. This runtime has to be smaller than $T_\textnormal{3\
steps} \cdot N_\textnormal{3\ steps}$ to make the reordering of steps pays
off.
$N_\textnormal{fused}$ is the number of time steps required by the fused
algorithm,
$N_\textnormal{3\ steps}$ is the number of time steps required by the
baseline.
To make the fusion pay off, the total time of the fused scheme has to be smaller
or equal to a straightforward implementation.
We obtain
\[
1 \leq
C_\textnormal{rerun} <
1+ \frac{ N_\textnormal{3\ steps}\cdot T_\textnormal{3\
steps}-N_\textnormal{fused} \cdot T_\textnormal{fused}} {N_\textnormal{fused}
\cdot T_\textnormal{STP}},
\]
with a trivial lower bound. We insert $N_\textnormal{fused} = \frac{1}{
C_{\Delta T} } N_\textnormal{3\ steps}$.
\end{proof}

\noindent
Corollary \ref{corollary:rerun-constraint} clarifies that
there exists a trade-off between the cost of the \texttt{STP} reruns and
the time step size damping:
The smaller we make $C_{\Delta T}$, the smaller
$C_\textnormal{rerun}$ but the larger $N_\textnormal{fused}$.
Balancing can be delicate.

Our gliding average in
Alg.~\ref{algorithm:algorithm-efficient-ADER-DG} induces
$
  N_\textnormal{fused} \geq \frac{1}{ C_{\Delta T} } N_\textnormal{3\ steps},
$
i.e.~we damp the effective $ \Delta T $ further. 
This inequality seems to harm the corollary's upper bound.
Yet, as long as the admissible time step size increases slowly or remains
invariant, the inequation holds trivially.
No reruns occur.
Occasional reductions of $\Delta T_\textnormal{adm}$ to no more
$C _{\Delta T}$ of the previous admissible time step size do not lead to reruns
and thus make the inequality hold more robustly.
The corollary quantifies how often the admissible time step size may decrease
more dramatically.

\begin{corollary}
\label{corollary:weak-single-touch}
 We assume that any individual tasks can complete in cache.
 We furthermore assume that data is completely removed from the cache if it is
 not reused a small, fixed number of algorithmic steps later by another task.
 The weak single-touch theorem alters the number of data reads
 and writes per time step by the ratio 
 \[
 0.875 \leq
 \frac{\left( 4\,d+2\right) \,C_\textnormal{rerun}+12\,d+1}{18\,d+4}
 \leq
 1.125.
 \]
\end{corollary}

\begin{proof}
 The \texttt{STP} reads in the old time step, and determines all data along
 the cell's hull (Corollary \ref{corollary:hull}).
 The Riemann solver reads in ${Q_h^*} ^{\pm}$ and $F(Q_h^*) ^{\pm} \cdot n$ and
 writes out $F_h^*$ into both $F(Q_h^*) ^{\pm} \cdot n$ arrays.
 Neglecting the domain boundary, 
 there are $d$ faces per cell.
 The \texttt{Corrector} reads in the face data---the \texttt{integrateVolume}
 already has been anticipated by the \texttt{STP}---and performs the remaining
 tasks. We have
 \begin{equation*}
 \begin{gathered}
  \underbrace{8dm(p+1)^d}_{\texttt{Riemann}\,\mbox{reads}}
  +
  \underbrace{4dm(p+1)^d}_{\texttt{Riemann}\,\mbox{writes}}
  + 
  \underbrace{2dm(p+1)^d+m(p+1)^d}_{\texttt{Corrector}\,\mbox{reads}}
  + 
  \underbrace{m(p+1)^d}_{\texttt{Corrector}\,\mbox{writes}}
  + \\
  \underbrace{m(p+1)^d}_{\texttt{STP}\,\mbox{reads}}
  + 
  \underbrace{4dm(p+1)^d+m(p+1)^d}_{\texttt{STP}\,\mbox{writes}}.
 \end{gathered}
 \end{equation*}
 
 \noindent
 Our fused scheme kicks off a grid traversal with the \texttt{Riemann} reads.
 We assume that a cache-oblivious or cache-aware traversal order is chosen
 \cite{bungartz_pde_2010}.
 In our case, we exploit the Peano space-filling
 curve \cite{bader_space_filling_2013}.
 \texttt{Riemann} writes thus
 remain in cache for the subsequent correction. 
 Hence, the \texttt{Corrector} solely has to read in the $Q_h^*$ values
 stored in the $D_h$ data structure to determine the next solution (written to memory). 
 This output is handed over to the next \texttt{Predictor} immediately, i.e.~the
 \texttt{Predictor} works in cache and does not have to reload data. 
 It writes out the cell's space-time hull.
 If the optimistic time stepping runs into a rerun, we have to reload the old
 time step and perform the prediction again:
 \begin{equation*}
 \begin{gathered}
  \underbrace{8dm(p+1)^d}_{\texttt{Riemann}\,\mbox{reads}}
  +
  \underbrace{4dm(p+1)^d}_{\texttt{Riemann}\,\mbox{writes}}
  +
  \underbrace{m(p+1)^d}_{\texttt{Corrector}\,\mbox{reads}}
  +
  \underbrace{m(p+1)^d}_{\texttt{Corrector}\,\mbox{writes}}
  +\\
  \underbrace{4dm(p+1)^d+m(p+1)^d}_{\texttt{STP}\,\mbox{writes}}
  +
  (C_\textnormal{rerun}-1) \cdot \left(
    \underbrace{m(p+1)^d}_{\texttt{STP}\,\mbox{reads}}
    +
    \underbrace{4dm(p+1)^d+m(p+1)^d}_{\texttt{STP}\,\mbox{writes}}
  \right).
 \end{gathered}
 \end{equation*}$~$
\end{proof}

\noindent
The memory pressure of fused time stepping compared to a straightforward version
reduces if $C_\textnormal{rerun}$ remains reasonably small:
we transfer data less frequently.
Our approach falls into the second rubric of 
communication-avoiding techniques, and, from a memory point of view, also
fits rubric five.
Both estimates are too pessimistic for global time stepping.
Here, it is convenient to run \texttt{integrateFace} directly after the
actual Riemann solver.
Explicit spanning (and storing) the Riemann solution $F_h^*$ over the time
span is not required.

\begin{corollary}
 \label{corollary:concurrency-homogenisation}
 Our fused and shifted approach homogenises the concurrency level and reduces
 memory access bursts. Yet, it reduces the overall concurrency level slightly. 
\end{corollary}

We assume that we traverse the grid $\mathbb{T}$ with multiple threads.
With multithreading, we have to ensure that no two faces in Algorithm
\ref{algorithm:algorithm-efficient-ADER-DG} are subject to the Riemann solver
concurrently or ``solved'' twice.

\begin{proof}
Running ADER-DG's three steps one after another means that we run into step
\texttt{STP} with concurrency level $|\mathbb{T}|$ and very high arithmetic
intensity, then perform the Riemann solves with concurrency $|d \cdot
\mathbb{T}|$, and finally run again a step with level $|\mathbb{T}|$. 
$\frac{ T_\textnormal{STP} }{ T_\textnormal{3\ steps} }$ quantifies the fraction
of the runtime spent in the predictor.
For a straightforward implementation, we thus obtain a time-averaged concurrency level of 
\[
  CL_{avg} \approx \frac{ T_\textnormal{STP} }{ T_\textnormal{3\ steps} } \cdot
  |\mathbb{T}| + (1-\frac{ T_\textnormal{STP} }{ T_\textnormal{3\ steps} })\, d \cdot |\mathbb{T}| 
    = \left( d+(1-d)\frac{ T_\textnormal{STP} }{ T_\textnormal{3\ steps} } \right) |\mathbb{T}|.
\]
Our shifted and fused variant exhibits an overall homogeneous concurrency level
of roughly $|\mathbb{T}|$: 
We merge all three algorithmic steps, i.e.~process all the cells
in parallel.
Per cell, we have to ensure all faces have been subject to the Riemann solve
first.
This induces only a brief startup cost however, since
arithmetically heavy \texttt{predict} tasks are
overlapped with the Riemann solves.
$CL_{avg} \geq |\mathbb{T}|$.
\end{proof}


\noindent
Our approach falls into the third rubric of communication-avoiding techniques
as bursts of data reads or writes over the bus are avoided.
We note that the difference in concurrency levels disappears for increasing
\texttt{STP} cost and that a high concurrency level for low-cost tasks does not
automatically induce high speedups.

\begin{corollary}
 In a non-overlapping domain decomposition,
 we can hide data exchange behind the traversal.
\end{corollary}

\noindent
The distributed-memory parallelisation of ADER-DG with message passing is
conceptionally simple.
We rely on a non-overlapping domain decomposition where the Riemann problems
along the parallel subdomain boundaries are solved redundantly.
Each rank sends out its projected $Q_h^*$ and $F(Q_h^*)\cdot n$ values to the
neighbouring ranks once the \texttt{extrapolate} task has terminated.
All other operations are rank-local.

\begin{proof}
We restrict to a machine model where the ADER-DG steps are synchronised and 
reiterate that \texttt{predict} is the
arithmetically most demanding task.
It implies that messages holding extrapolated data have, in a straightforward
implementation, up to $\min(T_\textnormal{STP},T_\textnormal{fused}-T_\textnormal{STP})$ time to arrive at their
destination rank, if we choose the task ordering optimally and manage to exchange all data non-blocking.
The first message sent out has a maximal arrival time (slack) of $T_\textnormal{STP}$. 
The last predicted value will the latest be required by the last Riemann
solve.
The time span in-between is bounded by $T_\textnormal{fused}-T_\textnormal{STP}$.

For our fused and shifted algorithm, we may assume that the slack given to a
message to run through the system is bounded by $T_\textnormal{fused}$ if we choose a
proper task ordering on each rank.
The statement holds as
\begin{eqnarray*}
  T_\textnormal{STP} \leq T_\textnormal{fused} \Rightarrow \min(T_\textnormal{STP},T_\textnormal{fused}-T_\textnormal{STP}) \leq
  T_\textnormal{fused}.
\end{eqnarray*}

\noindent
We could harm the right condition iff $T_\textnormal{STP}>\frac{1}{2}T_\textnormal{fused}$ and
$T_\textnormal{fused}-T_\textnormal{STP} > T_\textnormal{fused}$.
However,
$T_\textnormal{fused}-T_\textnormal{STP} > \frac{1}{2}T_\textnormal{fused} > T_\textnormal{fused}$ cannot hold.
\end{proof}

\noindent
Our approach falls into the fourth rubric of
communication-avoiding algorithms, as we reduce the criticalness of
 data communication.
Rubric five applies automatically.
Localisation of data transfer (rubric five)
is implicitly given by ADER-DG---as for any explicit scheme---as information
almost exclusively is transferred through the faces.

\section{Results}
\label{section:results}

We conducted our experiments on an Intel
E5-2650V4 (Broadwell) cluster with 24 cores per node.
They run at 2.4 GHz and are connected by Omnipath.
All hardware counters have been evaluated through Likwid \cite{Treibig:10:Likwid}.
Furthermore, we ran our experiments on KNL processors (Xeon Phi 7210) with 64
cores each.
They run at 1.30 GHz.
For the shared memory parallelisation, we rely on Intel's Threading Building
Blocks (TBB) while Intel MPI is used for the distributed memory parallelisation.
Intel's 2017 C++ compiler translated all codes.

%
%
Our runtime analysis focuses on compressible Euler equations with
$m=5, B=S=0$ and no $\delta$ impact in
\eqref{equation:introduction:PDE}.
We study
\[
\frac{\partial}{\partial t} Q
+
\nabla \cdot F(Q)
= 0
\  \mbox{with} \ 
Q = \begin{pmatrix}
\rho\\j\\\ E
\end{pmatrix},
F = 
\begin{pmatrix}
j\\
\frac{1}{\rho} j \otimes j + p I \\
\frac{1}{\rho} j \, (E + p)
\end{pmatrix},
\
p=0.4\left(E-\frac{1}{2\rho} (j,j) \right)
.
\]

\noindent
$\rho, E: \Omega \mapsto \mathbb{R}$ encode the scalar density and the 
energy, while the vector $j: \Omega \mapsto \mathbb{R}^d$ holds the medium's
velocities.
A pressure $p$ calibrates the whole setup.
$\otimes$ is the outer dot product.

For communication-avoiding algorithms, this is a challenging
setup as the PDE is simple. Its arithmetic intensity 
per ADER-DG step is low compared to more complicated PDEs such as seismic or
gravitational waves.
No implementational and communication flaws are hidden by computations.

\subsection{Hardware counter}

\begin{table}
\caption{
  Performance counters tracking the memory bus usage for a $27 \times 27 \times
  27$ grid on one node.
  Smooth initial conditions ensure that $C_\textnormal{rerurn}=1$.
  The upper part holds data from the straightforward implementation with three
  distinct steps, i.e.~$T=T_\textnormal{3\ steps}$.
  The lower part holds data from our fused approach ($T=T_\textnormal{fused}$).
  Both approaches follow Corollary \ref{corollary:hull}.
  The bandwidth (BW) is given as MB/s, the volume (Vol.) transferred is given in
  GB.
  All timings are time per time step with $[T]=s$.
  \label{table:hw-counters}
}
\begin{center}
\footnotesize
\begin{tabular}{c rrr rrr rrr}
&
\multicolumn{3}{c}{1 core} & 
\multicolumn{3}{c}{12 cores} &
\multicolumn{3}{c}{24 cores} \\
\cmidrule(lr){2-4}
\cmidrule(lr){5-7}
\cmidrule(lr){8-10}
$p$ & BW & Vol. & $T$ & BW & Vol. & $T$ & BW &
Vol. & $T$
\\
\midrule
3 & 1,357.44 & 30.59 & $1.36$ & 2,885.94 & 35.24 & $0.75$ & 4,639.00 & 
73.32 & $0.71$ \\
5 & 1,155.42 & 101.64 &$3.91$& 4,158.00 & 117.06 & $1.25$ & 6,377.94 & 
223.06 & $1.04$ \\
7 & 806.87 & 215.91 & $14.75$ & 5,695.60 & 285.40 & $2.24$ & 8,540.10 & 
520.33 & $1.83$ \\
9 & 483.04 & 487.98 & $29.20$ & 20,894.15 & 4,376.39 & $4.66$ & 30,938.36 &
4,716.02 & $3.76$ \\
\midrule
3 & 1,233.97 & 24.05 &$1.14$ & 3,645.78 & 31.58 & $0.39$ & 5,481.20 & 
71.92 & $0.39$ \\
5 & 861.10 & 80.49 &$4.18$& 5,931.70 & 110.40 &$0.68$& 8,403.62 & 
211.44 & $0.57$ \\
7 & 625.40 & 176.84 &$10.71$& 6,877.66 & 350.95 &$1.98$& 9,003.53 & 
621.64 & $1.50$ \\
9 & 429.35 & 434.20 &$25.96$& 17,525.00 & 4,619.24 & $4.80$ & 27,297.87 &
5,280.57 & $4.32$ \\
\end{tabular}
\end{center}
\end{table}

%
%

%
%
We start our experiments with studies of the hardware counters on a single core
(Table \ref{table:hw-counters}).
A straightforward implementation with three grid sweeps per ADER-DG 
time step acts as baseline.
It already stores solely the space-time cell hull.
Fusion of all ADER-DG steps into a single-touch code reduces the
data running through the memory bus to 0.79--0.89 of the baseline.
This reduction translates into diminished bandwidth
requirements.
Besides one outlier for $p=5$ where the bandwidth drops dramatically---we are
not able to identify the reason for this---the application of our techniques
speeds up the code.

If we use one socket or all cores of the node, our optimisations increase the amount of doubles transferred over the
bus once $p \geq 7$.
Parallel to this, the used bandwidth now increases.
Besides for $p=9$, fusion rather robustly speeds up the code; the poor 
scalability is a strong scaling effect and scalability overall is subject of
study next.
In all of our experiments, the L2 cache miss rate of the optimised code variant
is in the order of 2\% wich occasional outliers up to 4.92\%.
The L3 miss rate is 0.1\% at most.
While cache misses basically do not exist, we see a rapid increase in the L3
access bandwidth for the multicore experiments once we increase $p$ beyond 5.
With Stream TRIAD \cite{McCalpin:95:Stream} yielding around 17,774.3 MB/s on a
single core and 121,495.2 MB/s on all 24 cores, our code is not bandwidth-bound.
Besides for $p=3$, our optimisation techniques increase the MFLOPS/s by a factor
of 1.7--2. 
The growth in MFLOPS/s scales with $p$.

%
%
%

Corollary \ref{corollary:weak-single-touch} predicts for
$C_\textnormal{rerun}=1$ and $d=3$ a reduction of the memory transfer demands to
0.775. 
This matches our single core observations.
In general, the corollary predicts a reduction of the memory volume for 
$C_\textnormal{rerun}<1.5$ as long as we assume that the idea behind  
Corollary \ref{corollary:hull} has been applied.
If we made the comparison's baseline, i.e.~the three step implementation, hold
the complete space-time polynomials and fluxes too, the single-touch reordering
reduced the transferred data for all $C_\textnormal{rerun} \leq 2.0$.

The caveat in practice is Corollary \ref{corollary:rerun-constraint}.
For the present setups and a $C_{\Delta T}=0.99$, we experimentally obtain
maximum values of $C_\textnormal{rerun}$ between 1.21 ($p=3$) down to 1.01
($p=7$) and 0.98 ($p=9$) on a single core.
These values demonstrate that reruns have to be avoided the more rigorously the
higher the cost of the predictor.
Yet, naively reducing the safety value $C_{\Delta T}$ does not work either.
If $C_{\Delta T}$ is chosen too pessimistic, i.e.~too small, the fusion even
might not pay off at all.
The optimistic time stepping requires too many
additional steps through the safety factor (cmp.~data for $p=9$).
To make our approach perform, a user has to balance 
$C_{\Delta T}$ and ensure that the number of reruns remains below a problem- and machine-specific
threshold.
Fortunately, this upper threshold on $C_\textnormal{rerun}$ increases if we
increase the core count as $T_\textnormal{STP}$ is almost embarrassingly scaling.

%
%
The merging and reordering of the algorithmic steps in combination with an
SFC-based grid traversal which is spatially and temporally localised \cite{Weinzierl:17:Peano}
plus the storage of solely the space-time hull yields a cache-oblivious code.
The MFLOPS/s increase due to the fusion of steps and the optimistic
time stepping.
Though our techniques are designed to optimise the memory access, 
they allow the compiler to exploit the vector facilities better;
even for our PDE with low arithmetic intensity.
They team up with the common knowledge that we have to increase the polynomial order to
increase the arithmetic intensity.

The optimisations' homogenisation of the concurrency level and the storage of
solely the cell hull are a double-edged sword. 
\texttt{STP} temporarily has to hold all space-time unknowns per cell. 
If the code runs in parallel, each thread has to allocate these temporary data.
The memory demands scale with both $p$ and the core count.
Though the code exhibits advantageous cache access characteristics, some data
now is evacuated from the high-level and last-level caches and then brought in
again.
For $p=9$, solely $T_\textnormal{STP}$ dominates the runtime.
Fusing tasks and thus intermixing and nonhomogenising memory accesses here makes
the code become inferior to a non-fused approach where \texttt{STP} streams data
with quasi-uniform operations through the cores.
With increasing last level cache sizes as they come along with an KNL, e.g.,
this performance ``anomaly'' will disappear.

\subsection{Shared memory scalability}

\begin{figure}
 \centering
   \includegraphics[scale=0.89]{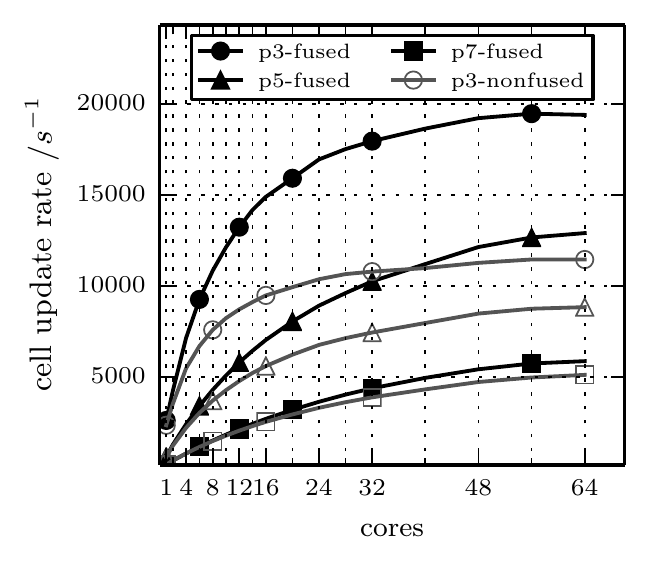}
   \includegraphics[scale=0.89]{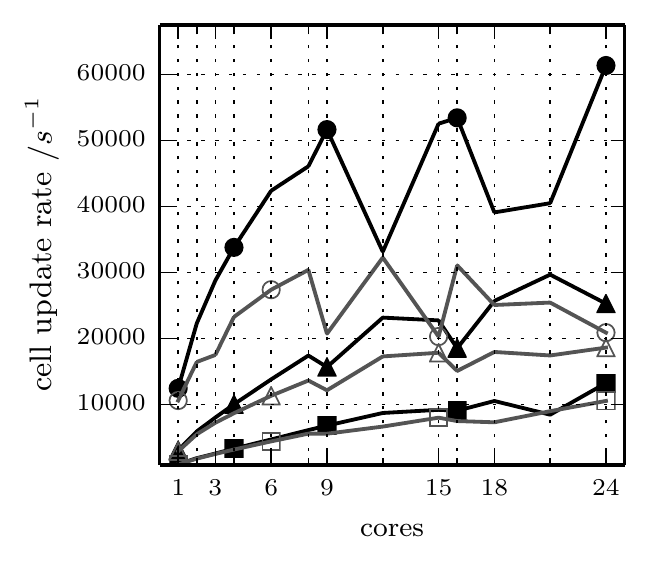}  
 \caption{Characteristic shared memory profiles for $d=3$ on KNL (left) and Broadwell
  (right).
  \label{figure:shared-mem}
 }
\end{figure}

%
%
We continue with scalability tests on single nodes
(Fig.~\ref{figure:shared-mem}) where we apply
sufficiently smooth initial conditions plus a $C_{\Delta T}$ to ensure that
$C_\textnormal{rerun}=1$.
We omit rerun effects.
%
%
Our techniques yield, for most of the setups,
better runtimes than their straightforward counterpart using three grid sweeps
for the three algorithmic steps.
On the multicore chip, the data are ragged while the manycore yields 
smooth curves.
Both curves show a widening speed gap, i.e.~the more cores we use the more
important our optimisations.
Yet, the higher the polynomial degree $p$, the less important our
optimisations.

%
%
A straightforward implementation exhibits three synchronisation
points per realisation time step with two of the phases in-between
(\texttt{Riemann} and \texttt{Corrector}) being cheap.
Relative to \texttt{STP}, they become almost serial tiny phases
for growing polynomial order.
Amdahl's law clarifies that the two cheap phases tend to throttle the
upscaling.
Our techniques eliminate two out of three synchronisation points per realisation
time step.
This way, they make the code scale better.
The gap between the two implementation variants widens.
Yet, the scaling improvement is constrained by the reduced concurrency level
from Corollary \ref{corollary:concurrency-homogenisation}.
The closer $T_\textnormal{STP}$ to $T_\textnormal{3 steps}$, i.e.~the higher
$p$, the harsher the reduction of the concurrency level.
The relative gain in performance through the optimisation techniques reduces
with growing polynomial degree.

%
%
On both architectures, our techniques pay off most for the smaller
$p$ choices---in particular for Finite Volumes---for tiny
time step sizes where few Picard iterations are required, and for linear variants of (\ref{equation:introduction:PDE}) where the
cost ill-balance between \texttt{STP} and other phases is not too significant.

\subsection{Distributed memory characteristics}

\begin{figure}
 \centering
   \includegraphics[scale=0.89]{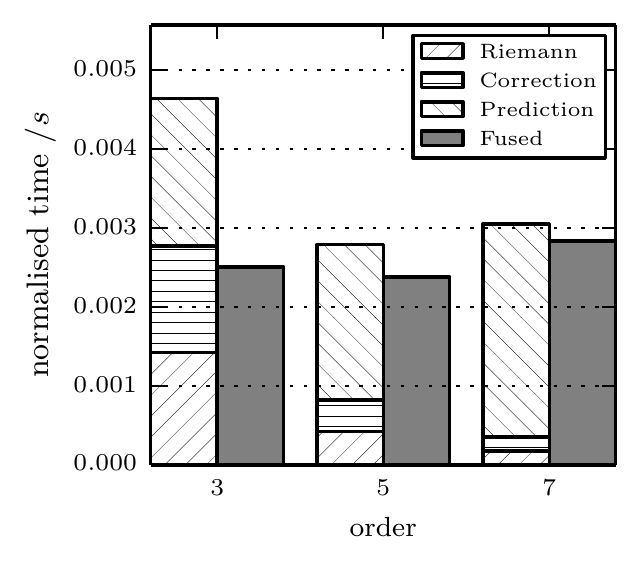}
   \includegraphics[scale=0.89]{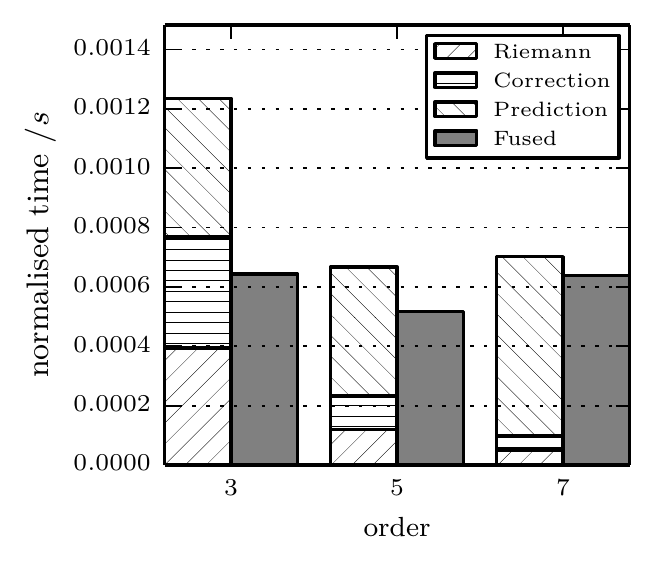}
 \caption{Time per time step for a $81 \times 81 \times 81$ mesh on KNL
 (left) and Broadwell (right). 28 ranks are used.
  The time is normalised w.r.t. the $81^3 \cdot (p+1)^3$ Legendre points of the
  solution.
  \label{figure:distributed-mem}
 }
\end{figure}

We close the experiments with distributed memory tests realised through MPI. 
Shared memory parallelisation is disabled to avoid that different 
runtime effects interfere.
All experiments are chosen such that the analysis covers only simulation spans
where the work is reasonably good balanced.
No ill-balancing pollutes the data.
Finally, the setup is chosen such that no reruns are observed.
Fused time stepping continues to pay off throughout all of our experiments
(cmp.~exemplary Fig.~\ref{figure:distributed-mem}).
The pay off however is limited on the manycore.
Effectively, the measured time for the fused approach equals the time required
for the sole \texttt{STP} grid sweep.
Given the massive growth of \texttt{STP} cost for increasing polynomial degrees,
fusion's impact is more significant for small $p$.

ADER-DG's prediction triggers data exchange through the cell faces.
This data has to arrive for the subsequent Riemann solves whereas no data
transfer is required when we progress from the Riemann solutions to the
correction step.
All results suggest that all the data exchange is successfully hidden behind the
computation.
The predictor dominates the parallel runtime, but almost
all data exchange happens in the background of this \texttt{STP} as we do not
see a significant increase of the \texttt{Riemann} timings.
The improvement characteristics of our fused, communication-avoiding ADER-DG
variant are preserved for distributed memory runs.

Two insights put this success into perspective:
The hiding of communication applies for the non-fused
realisation, too.
It is solely the immediate firing of Riemann input data by the \texttt{STP}
tasks that makes the data transfers hiding behind computations.
The other techniques have no positive impact.
Instead, only lower $p$ orders yield $T_{fused} + T_{STP} < T_{3\
steps}$.
The optimistic time stepping remains not without risk for parallel runs,
i.e.~have to be avoided for higher polynomial order.

\section{Conclusion}
\label{section:conclusion}

Our manuscript details the implementation of an ADER-DG code that exhibits
communication-avoiding characteristics.
For this, we generalise and detailed the term communication-avoiding
itself.
Our experimental data with ADER-DG employing a Finite Volume limiter plus our
analyses suggest that the techniques pay off.
The techniques or variations of them apply and are of use for a vast range of
predictor-corrector schemes in general.

There are two natural extensions of the present research.
They are subject to ongoing work.
On the one hand, we have to demonstrate the impact and usefulness of the
presented techniques for real-world simulation challenges.
This is part of the ExaHyPE \cite{Software:ExaHyPE} agenda.
On the other hand, we have to detail how grid layouts and implicitly
generated task graphs as well as grid traversals and implicitly defined task
graph processing interplay.
This is important for real-world scaling.
Our results suggest that tasks aligning along subdomain boundaries in an MPI
context are critical, i.e.~have to be handled as soon as possible.
Prioritised processing allows the
corresponding Riemann input data to squeeze through the network while further
tasks are processed.
Also, we do present MPI results for given setups which illustrate the
effectiveness of the proposed solutions.
Appropriate load balancing has been enforced manually.
In practice, the quality of the load balancing in context with the tasking will
determine the achieved performance.

\section*{Acknowledgements}

The authors appreciate support received from the European Union’s Horizon 2020
research and innovation programme under grant agreement No 671698 (ExaHyPE).
This work made use of the facilities of the Hamilton HPC Service of Durham
University, and it used the ARCHER UK National Supercomputing Service
(http://www.archer.ac.uk), i.e.~its KNL nodes.
Thanks are due to all members of the ExaHyPE consortium who made this research
possible.
Particular assistance with the numerical scheme has been provided by Michael
Dumbser.
%
All underlying software is open source \cite{Software:ExaHyPE}.

\bibliographystyle{siam}
\bibliography{efficient-aderdg}{}

\appendix
\section{Finite Volume schemes within the predictor-corrector formalism}
\label{appendix:aderdg}

\begin{corollary}
 \label{corollary:FV-Rusanov}
Finite Volumes with Rusanov fluxes can be cast into a predictor-corrector
scheme. 
\end{corollary}

\begin{proof}
Let $p=0$ and let the predictor be the identity, i.e.~the predicted solution
$Q^*_h(T+\Delta T) = Q_h(T)$: 
We make \texttt{extrapolate} write $Q^*_h$ and a $F(Q^*_h)$ onto the $2d$
faces of each cell and write $id_{c \mapsto \partial c}$
With this information, \texttt{Riemann} determines
$\frac{1}{2}(Q^{+}_h+Q^{-}_h)+\alpha \cdot \left(F^{+}_h-F^{-}_h \right)$.
We write \texttt{solveRiemann}=\texttt{computeRusanovFluxes}.
As $p=0$, both integration tasks multiply their input with $\Delta T$,
i.e.~(\ref{equation:operator}) becomes
\[
  Q_h (T+\Delta T) =  
   \left(id +
   \Delta T \cdot
  \texttt{computeRusanovFluxes}
  \circ id_{c \mapsto \partial c} 
  \right)
  Q_h (T). 
\]
\end{proof}

\noindent
The result is not surprising as we constructed our ADER-DG data structures to
support (space-time) Rusanov solvers.
Yet, the corollary illustrates how the limiter is to be realised:
Here we embed $(2d+1)^d$ patches into each ADER-DG cell.
The subdivision, i.e.~the patches, compensate for the high order.
As patches are equivalent to $p=0$ ADER-DG on a subgrid, the faces of our
limiter hold
$(2p+1)^{d-1}$ $Q$/$F$-tupels as opposed to $(2p+1)^{d}$ space-time entries.
The Finite Volume scheme has lower total memory demands, but otherwise allows
for exactly the same data flow and algorithmic step paradigms.

\begin{corollary}
Finite Volumes with a MUSCL-Hancock 
Riemann solvers can be cast into our ADER-DG predictor-corrector
formalism if we (i) restrict to tensor-product style operators, i.e.~neglect
interaction of a Finite Volume cell with neighbours that are not face-connected,
and (ii) rely on patches of at least two Finite Volume cells per coordinate axis per
cell.
\end{corollary}

\begin{proof}
We preserve the constant extrapolation in time in the predictor from Corollary
\ref{corollary:FV-Rusanov}.
While we project $Q^*_h=Q_h$ onto the face, we do not map the corresponding flux
onto the face as we observe that $Q^*_h=Q_h$ implies that these fluxes hold
redundant information anyway.
Instead we project the gradients (weighted difference) of $Q_h(T)$ to the faces.

With the second entry on the face, i.e.~the $F$ entry, holding 
$(\nabla Q_h^*,n)^{-}$ and $(\nabla Q_h^*,n)^{+}$, we can reconstruct
the $Q$ values of neighbour of neighbour cells along a coordinate axis.
Indeed, it is reasonable to technically store those values directly in $F$
instead of gradients along a normal.
Knowing neighbours of neighbours allows us to realise MUSCL-Hancock's internal
predition for half time steps with the \texttt{Riemann} task.
The overall scheme follows the proof of Corollary
\ref{corollary:FV-Rusanov} with
\texttt{solveRiemann}=\texttt{computeMUSCLHancockFlux}.
\end{proof}

\section{Proof of Corollary \ref{corollary:naive-element-wise-single-touch}}
\label{appendix:theorems}

\begin{proof}
ADER-DG's minimalist version (without limiter and adaptive mesh refinement,
e.g.) runs three steps per time step which decompose into individual tasks.
\texttt{STP} and \texttt{Corrector} are defined on one cell, while
\texttt{Riemann} accepts input data from two cells.
Explicit time stepping means that ADER-DG runs a sequence of steps 
$\texttt{STP} \before \texttt{Riemann} \before \texttt{Corrector} \before
\texttt{STP} \before \texttt{Riemann} \before \ldots$.
Element-wise single-touch means that there is a splitting of this sequence
into chunks of three tasks such that the data $Q_h$ of any cell is read and written
only once per task triad.
We assume that such a splitting exists.
The periodicity of the sequence implies that we have to analyse three variants
to find it.

We first assume that we can split up ADER-DG into $(\texttt{STP} \before
\texttt{Riemann} \before \texttt{Corrector})^+$ step sequences.
Let $c_a$ be subject to \texttt{STP}. 
We have to assume that it yields the input to \texttt{Riemann} for one
face while the data from the face-connected neighbouring cell $c_b$ is already
available. 
Consequently $c_b \before c_a$. 
To ensure a single touch, \texttt{Corrector} on $c_a$ follows immediately to the
Riemann solve.
This means the code is not single-touch for $c_b$.

We second assume that we can split up ADER-DG into $\texttt{STP} \before 
(\texttt{Riemann} \before \texttt{Corrector} \before \texttt{STP})^+$ step
sequences.
The single touch constraint here is rewritten from a cell-based notion into a 
face-based notion, i.e.~face information is used to keep multiple time steps
consistent.
To ensure single touch, we have for any cell $c_a$ to run
\texttt{STP} directly after \texttt{Corrector}. 
\texttt{Corrector} comprises a task \texttt{calcTimeStepSize} (Table
\ref{table:memory-requirements}).
Without further assumptions on the time stepping scheme, i.e.~that the solution
evolves smoothly and $\Delta T_{adm}$ thus is continuously increasing, a cell
$c_b$ with $\texttt{STP}(c_a) \before \texttt{STP}(c_b)$ might reduce the
admissible time step used for $c_a$. 
We thus have to recompute $\texttt{STP}(c_a)$ and are not single touch anymore.

We finally assume that we can split up ADER-DG into $\texttt{STP} \before
\texttt{Riemann} \before (\texttt{Corrector} \before \texttt{STP} \before
\texttt{Riemann})^+$ step sequences.
While this is indeed single-touch w.r.t.~ cell data, the previous argument on
the exchange of admissible time step sizes $\Delta T_{adm}$ continues to hold.

Our assumption that there is a single-touch element-wise algorithm has been
wrong.

\end{proof}
\section{Technical details on the properties}

\paragraph{ Proof for Corollary \ref{corollary:hull} }

A straightforward implementation of ADER-DG exhibits at least the following
persistent memory footprint per cell following Table
\ref{table:memory-requirements}.
\texttt{STP} works on the following data cardinalities:
\[
 \underbrace{
  m(p+1)^d
 }_{
  Q_h
 }
 +
 \underbrace{
  (d+1)m(p+1)^{d+1}
 }_{
  \mbox{from} \ \texttt{predict}
 }
 +
 \underbrace{
  4dm(p+1)^d
 }_{
  \mbox{from} \ \texttt{extrapolate}
 }
\]

\noindent
The Riemann solve yields $F^*_h$ over space-time though space here refers to
cell faces only.
The result is to be stored in a separate data structure.
Again, all quantities are normalised w.r.t.~cell count:

\[
 \underbrace{
   2dm(p+1)^d
 }_{
  \mbox{from} \ \texttt{solveRiemann}
 }
\]

\noindent
The actual update can be performed in-situ in $Q_h$. 
Time step size computations do not increase the memory footprint which totals in
\[
  (d+1)m(p+1)^{d+1} + (1+6d)m(p+1)^d.
\]

\noindent
The present paper's variant induces, according to Table
\ref{table:memory-requirements}, a memory footprint of
\[
 \underbrace{
  2m(p+1)^d
 }_{
  Q_h\ \mbox{and}\ D_h
 }
 +
 \underbrace{
  4dm(p+1)^d
 }_{
  \mbox{from} \ \texttt{extrapolate}
 }
 +
 \underbrace{
   2dm(p+1)^d
 }_{
  \mbox{from} \ \texttt{solveRiemann}
 }
 = 
 (2+6d)m(p+1)^d.
\]

\noindent
We obtain a relative memory footprint of our implementation compared to a
straightforward, 3-step code of

\begin{center}
 \begin{tabular}{r|rr}
 $p$ & $d=2$ & $d=3$ \\
 \hline
 2 & 0.64 & 0.65 \\
 3 & 0.56 & 0.57 \\
 4 & 0.50 & 0.51 \\
 5 & 0.45 & 0.47 \\
 6 & 0.41 & 0.43 \\
 7 & 0.38 & 0.39 \\
 8 & 0.35 & 0.36 \\
 9 & 0.33 & 0.34
 \end{tabular}
\end{center}

\paragraph{ Remarks on Corollary \ref{corollary:concurrency-homogenisation}}

The concurrency models in Corollary \ref{corollary:concurrency-homogenisation}
lack the fact that all three phases (in the straightforward variant) or the
fused approach, respectively, contain tiny serial fragments plus startup cost. 
These fragments materialise in limited scalability in our results.

The inequality itself follows trivially from
\begin{eqnarray*}
  d + (1-d) \frac{ T_\textnormal{STP} }{ T_\textnormal{3\ steps} } \Rightarrow 
  2 - \frac{ T_\textnormal{STP} }{ T_\textnormal{3\ steps} } & \geq & 1 
  \qquad \mbox{for}\ d = 2,
  \\
  3 - 2\frac{ T_\textnormal{STP} }{ T_\textnormal{3\ steps} } & \geq & 1
  \qquad \mbox{for}\ d = 3
\end{eqnarray*}

\end{document}